\documentclass[submission,copyright,creativecommons]{eptcs}
\providecommand{\event}{NCMA 2023} 

\usepackage[english]{babel}

\usepackage{latexsym}
\usepackage{amsmath}
\usepackage{amsthm}
\usepackage{amsfonts}
\usepackage{amssymb} 
\usepackage{comment}
\usepackage{braket}
\usepackage{shuffle}
\usepackage{caption}
\usepackage{subcaption}

\usepackage{tikz}
\usepackage{tikz-qtree}
\usepackage{tikz-qtree-compat}
\usetikzlibrary{shapes.geometric}
\usetikzlibrary{shapes,arrows,decorations.pathmorphing,backgrounds,positioning,fit,petri,patterns, calc}
\usepackage{pgfplots}

\newcommand{\e}{{\mbox{{\bf\em e}}}}
\renewcommand{\bigcirc}{{\circlearrowleft}}

\newcommand{\R}{{\mathbb{R}}}

\newcommand{\C}{{\mathbb{C}}}
\newcommand{\A}{{\mathcal A}}
\renewcommand{\O}{\mathcal{O}}

\renewcommand{\mod}{{\rm \;mod\;}}

\renewcommand{\r}{\sharp}

\newcommand{\sett}[2]{{\left\{#1 \ \mid \ #2\right\}}}
\newcommand{\nradix}[1]{{{\rm e}^{i\frac{2\pi}{n}{#1}}}}
\newcommand{\dfa}{\textsc{dfa}}
\newcommand{\dfas}{\textsc{dfa}s}
\newcommand{\qfa}{\textsc{qfa}}
\newcommand{\qfas}{\textsc{qfa}s}
\newcommand{\lqfa}{\textsc{lqfa}}
\newcommand{\lqfas}{\textsc{lqfa}s}

\newcommand{\mo}{\textsc{mo-qfa}}
\newcommand{\mos}{\textsc{mo-qfa}s}
\newcommand{\mm}{\textsc{mm-qfa}}
\newcommand{\mms}{\textsc{mm-qfa}s}

\newcommand{\norm}[1]{\left\|#1\right\|}

\newcommand{\inner}[2]{{\left\langle#1,#2\right\rangle}}

\newcommand{\cc}[1]{{#1}^*}

\newtheorem{thm}{Theorem}
\newtheorem{definition}{Definition}

\ifpdf
  \usepackage{underscore}         
  \usepackage[T1]{fontenc}        
\else
  \usepackage{breakurl}           
\fi

\title{Latvian Quantum Finite State Automata\\for Unary Languages}
\author{Carlo Mereghetti\qquad\qquad Beatrice Palano\qquad\qquad Priscilla Raucci
\institute{Dipartimento di Informatica ``Giovanni Degli Antoni''\\
Universit\`a degli Studi di Milano, via Celoria 18, 20135 Milano -- Italy}
\email{\{carlo.mereghetti, beatrice.palano, priscilla.raucci\}@unimi.it}
}

\newcommand{\titlerunning}{Latvian Quantum Finite State Automata for Unary Languages}
\newcommand{\authorrunning}{C. Mereghetti, B. Palano \& P. Raucci}

\hypersetup{
  bookmarksnumbered,
  pdftitle    = {\titlerunning},
  pdfauthor   = {\authorrunning},
  pdfsubject  = {EPTCS},               
  pdfkeywords = {Descriptional complexity, Quantum finite state automata, Unary languages} 
}

\begin{document}
\maketitle

\begin{abstract}
We design \emph{Latvian quantum finite state automata} (\lqfas\ for short) recognizing unary regular languages with isolated cut point $\frac{1}{2}$.
From an architectural point of view, we combine two \lqfas\ recognizing with isolated cut point, respectively, the finite part and the ultimately periodic part of any given unary regular language~$L$. In both modules, we use a component addressed in the literature and here suitably adapted  to the unary case, to discriminate strings on the basis of their length. The number of basis states and the isolation around the cut point of the resulting \lqfa\ for $L$ exponentially depends on the size of the minimal deterministic finite state automaton for $L$.
\end{abstract}
\section{Introduction}\label{sec:intro}
Quantum finite automata (\qfas\ for short) represent a theoretical model for a quantum computer with finite memory
\cite{BMP14,BMP17}. While we can hardly expect to see a full-featured quantum computer in the near future, small quantum components, modeled by \qfas, seem to be promising from a physical implementation viewpoint (see, e.g., \cite{CMP21,MPC20}).

Very roughly speaking, a \qfa\ is obtained by imposing the quantum paradigm --- superposition, unitary evolution, observation  ---
to a classical finite state automaton. The state of the \qfa\ can be seen as a linear combination of classical states, called superposition. The \qfa\ steps from a superposition to the next one by a unitary (reversible) evolution. Superpositions can transfer the complexity of a computation from a large number of sequential steps to a large number of coherently superposed classical states (this phenomenon is sometimes referred as {quantum parallelism}). Along its computation, the \qfa\ can be ``observed'', i.e., some features, called observables, can be measured. From measuring an observable, an outcome is obtained with 
a certain probability and the current superposition irreversibly ``collapses'', with the same probability, to a particular superposition (coherent with the observed outcome).

\qfas\ exhibit both advantages and disadvantages with respect to their classical (deterministic or probabilistic) counterpart. Basically, quantum superposition offers some computational advantages on probabilistic superposition. On the other hand, quantum dynamics are reversible: because of limitation of memory, it is sometimes impossible to simulate deterministic finite state automata (\dfas\ for short) by quantum automata. Limitations due to reversibility can be partially attenuated by systematically introducing measurements of suitable observables as computational steps.

In the literature, several models of \qfas\ are proposed, which mainly differ in their measurement policy.
The first and most simple model is the {\em measure-once \qfa} (\mo\ for short) \cite{brodsky2002characterizations,MC00}, where
the probability of accepting strings is evaluated by ``observing'' just once, at the end of input processing. In
{\em measure-many} \qfas\ (\mms\ for short) \cite{kondacs1997power}, instead, the acceptance probability is evaluated by observing
after each move, thus allowing the possibility of halting the computation in the middle of input processing.
An additional model is the {\em Latvian \qfa} (\lqfa\ for short) \cite{Ambainis06}, which can be regarded as ``intermediate'' between \mos\ and \mms. In fact, as in the \mm\ model, \lqfas\ are observed after each move; on the other hand, as in the \mo\ model, acceptance probability is evaluated at the end of the computation only. From a language recognition point of view, it is well known that \mos\ are strictly less powerful than \lqfas, which are strictly less powerful than \mms, which are strictly less powerful than \dfas. This hierarchy is established, e.g., in \cite{Ambainis06,brodsky2002characterizations,kondacs1997power}.
\smallskip

In this paper, we investigate the architecture and size of \lqfas\ processing \emph{unary languages}, i.e. languages 
built over a single-letter alphabet. 
A similar investigation is presented in~\cite{Bianchi10}, where \mms\ recognizing unary regular languages with isolated cut point are exhibited, whose size (number of basis states) is linear in the size of equivalent minimal \dfas. Here, we show that unary regular languages can be recognized with isolated cut point by the less powerful model of \lqfas\ as well, paying by an exponential size increase. A relevant module in our construction is a \lqfa\ 
recognizing with isolated cut point the strings of length exceeding a fixed threshold. For its design, we adapt a construction in \cite{Ambainis06,Mercer07} to the unary case. Such a module is then suitably combined with two \lqfas\ taking care, respectively, of the finite part and the 
ultimately periodic part any unary regular language consists of. The architecture of the resulting \lqfa\ turns out to be significantly different from the equivalent \mms\ in~\cite{Bianchi10}. Moreover, while in the \mm\ case the isolation around the cut point is constant, for \lqfas\ it exponentially decreases with respect to the size of the \dfa\ for the finite part of the target unary regular language. However, it should be stressed that the less powerful model of \mos\ cannot recognize with isolated cut point all unary regular languages. Our results constructively prove that \lqfas\ and \mms\ have the same recognition power, whenever restricted to recognize unary languages with isolated cut point.
\smallskip

The paper is organized as follows. 
In Section \ref{sec:prel}, we overview basics on formal language theory, linear algebra, and quantum finite state automaton models. In Section \ref{sec:Mh_description}, we design isolated cut point \lqfas\ recognizing 
the strings whose length is greater than or equal to a fixed value. 
Then, in Section \ref{sec:unary_construction}, we provide the full architecture of isolated cut point \lqfas\ for unary regular languages, analyzing their size, cut point, and isolation.
Finally, in Section \ref{sec:conclusion}, we draw some concluding remarks and offer possible research hints.

	\section{Preliminaries}\label{sec:prel}
	
	\subsection{Formal Languages}\label{lang}
	We assume familiarity with basic notions of formal language theory (see, e.g., \cite{hopcroft2001introduction}).
Given a set $S$, we let $|S|$ denote its cardinality.
	The set of all words or strings (including the empty string $\varepsilon$) over a finite alphabet $\Sigma$ is denoted by $\Sigma^*$, and we let $\Sigma^+=\Sigma^*\setminus{\varepsilon}$. 
		For a string $w \in \Sigma^*$, we let $|w|$ denote its length and $w_i$ its {\em i}th symbol.
	For any given $i\ge0$, we let~$\Sigma^i$ be the set of strings over $\Sigma$ of length $i$, 
	with $\Sigma^0 = \{ \varepsilon\}$. 
	We let $\Sigma^{\leq i} = \bigcup_{j=0}^{i} \Sigma^j$; sets~$\Sigma^{>i}$ and $\Sigma^{\ge i}$ are defined accordingly.
	A language over $\Sigma$ is any subset $L\subseteq\Sigma^*$; its complement is the language $L^c=\Sigma^*\setminus L$.
	A~{\em deterministic finite state automaton} (\dfa) is formally defined as a 5-tuple $D = (Q, \Sigma, q_0, \delta, F)$,
	where~$Q$ is the finite set of states, $\Sigma $ the finite input alphabet, $q_0 \in Q$ the initial state, $F \subseteq Q$ the set of accepting states, and $\delta: Q \times \Sigma \rightarrow Q$ the transition function. 
	Denoting by $\delta^*$ the canonical extension of $\delta$ to~$\Sigma^*$, the language recognized by $D$ is the set $L_D = \{ w \in \Sigma^* | \delta^*(q_0, w) \in F\}$.
It is well known that \dfas\ characterize the class of regular languages.
\smallskip
	
	A {\em unary language} is any language built over a single letter alphabet, e.g., $\Sigma = \set{\sigma}$, and thus has the general form $L \subseteq \sigma^*$.  Unary {\em regular} languages form {\em ultimately periodic sets}, as stated by the following
	
	\begin{thm} \label{chrobak}
		\emph{(\cite{hopcroft2001introduction,Par66})} Let $L \subseteq \sigma^*$ be a unary regular language. 
		Then, there exist two integers $T \geq 0$ and $P > 0$ such that, for any $k \geq T$, we have $\sigma^k \in L$ if and only if $\sigma^{k+P} \in L$. 
	\end{thm}
	
	By Theorem \ref{chrobak}, it is easy to see that any unary regular language $L$ can be recognized by a (minimal) $\dfa$ consisting of an initial path of $T$ states joined to a cycle of $P$ states; accepting states are suitably settled on both the path and the cycle. Unary regular languages satisfying Theorem \ref{chrobak} with $T=0$ are called {\em periodic languages} of period $P$.
	
	\subsection{Linear Algebra}\label{lin}

	We quickly recall some notions of linear algebra, useful to describe quantum computational devices. For more details, we refer the reader to, e.g., \cite{shilov1971linear}.
	The fields of real and complex numbers are denoted by~$\R$ and $\C$, respectively.
	Given a complex number $z=a+ib$, with $a,b\in\R$, its {\em conjugate} is denoted by $\cc{z}=a-ib$, and its {\em modulus} by $|z|=\sqrt{z\cdot\cc{z}}$.
	We let $\C^{n\times m}$ denote the set of $n\times m$ matrices with entries~in~$\C$.  
	Given a matrix $M\in\C^{n\times m}$, for $1\le i\le n$ and $1\le j\le m$, we denote by~$M_{ij}$  its $(i,j)$th entry. 
	The {\em transpose} of $M$ is the matrix $M^T\in \C^{m\times n}$ satisfying ${M^T}_{ij}=M_{ji}$, while we let~$\cc{M}$ be the matrix satisfying ${\cc{M}}_{ij}=\cc{(M_{ij})}$. The {\em adjoint} of $M$ is the matrix $M^\dag=\cc{(M^T)}$. 
	For matrices $A,B\in\C^{n\times m}$, their~\emph{sum} is
	the $n\times m$ matrix $(A+B)_{ij}=A_{ij}+B_{ij}$. For matrices $C\in\C^{n\times m}$ and
	$D\in\C^{m\times r}$, their \emph{product} is the~$n\times r$ matrix $(C\cdot D)_{ij}=\sum_{k=1}^{m}
	C_{ik}\cdot D_{kj}$.  
	For matrices $A\in\C^{n\times
		m}$ and $B\in\C^{p\times q}$, their \emph{direct (or tensor or Kronecker)
		product} is the $n\!\cdot\!\! p\times m\!\cdot\!q$ matrix defined as
	$$
	A\otimes B=\left(\begin{array}{ccc}
		A_{11}B & \cdots & A_{1m}B\\
		\vdots  & \ddots & \vdots\\
		A_{n1}B & \cdots & A_{nm}B
	\end{array}
	\right).
	$$
	When operations are allowed by matrix dimensions, we have that $(A\otimes B)\cdot(C\otimes D)= A\cdot C \;\otimes\;B\cdot D$.
	
	A \emph{Hilbert space} of dimension $n$ is the linear space $\C^{n}$ of $n$-dimensional complex row vectors equipped with sum and product by elements in $\C$, where the {\em inner product} $\inner{\varphi}{\psi}=\varphi\cdot\psi^\dag$ is defined, for vectors $\varphi,\psi\in\C^n$.
	The $i$th component of $\varphi$ is denoted by $\varphi_i$, its  {\em norm} is given by $\norm{\varphi}=\sqrt{\inner{\varphi}{\varphi}}=\sqrt{\sum_{i=1}^n{|\varphi_i|}^2}$.
	If
$\inner{\varphi}{\psi}=0$ (and $\norm{\varphi}=1=\norm{\psi}$), then $\varphi$ and~$\psi$ are {orthogonal} ({orthonormal}\/). An orthonormal basis of $\C^{n}$ is any set of $n$ orthonormal vectors in $\C^n$. In particular, the {\em canonical basis} of $\C^{n}$ is the set $\set{{\e}_1,\e_2,\ldots,\e_n}$, where $\e_i\in\C^n$ is the vector having 1 at the $i$th component and 0 elsewhere. 
Clearly, any vector $\varphi\in\C^n$ can be univocally expressed as a linear combination of the vectors in the canonical basis as $\varphi=\sum_{i=1}^n\varphi_i\cdot\e_i$. This latter fact is usually addressed by saying that $\C^n$ is {\em spanned} by $\set{{\e}_1,\e_2,\ldots,\e_n}$.
Two subspaces~$X,Y\subseteq\C^{n}$ are
orthogonal if any vector in $X$ is orthogonal to any vector in $Y$.
In this case, we denote by $X\dotplus Y$ the linear space generated by $X\cup Y$.
	 For vectors $\varphi\in\C^{n}$
and $\psi\in\C^{m}$, 
their {direct (or tensor or Kronecker) product} is the vector
$\varphi\otimes\psi =(\varphi_1\cdot\psi,\ldots,\varphi_n\cdot\psi)$;
we have $\norm{\varphi\otimes\psi}=\norm{\varphi}\cdot\norm{\psi}$.

	A matrix $M\in\C^{n\times n}$ is said to be {\em unitary} if $M\cdot M^\dag=I^{(n)}=M^\dag\cdot M$, where~$I^{(n)}$ is the $n\times n$ identity matrix.
	Equivalently, $M$ is unitary if it preserves the norm, i.e., $\norm{\varphi \cdot M}=\,\norm{\varphi}$ for any vector $\varphi\in\C^{n}$. Direct products of unitary matrices are unitary as well.
	The matrix $M$~is said to be {\em Hermitian (or self-adjoint)} if $M=M^\dag$. 
	Let $\O\in\C^{n\times n}$ be an Hermitian matrix, $\nu_1,\nu_2,\ldots, \nu_s$ its eigenvalues, and $E_1, E_2, \ldots, E_s$ the corresponding eigenspaces. 
	It is well known that each eigenvalue~$\nu_k$ is real, that $E_i$ is orthogonal to~$E_j$ for every $1\le i\neq j\le s$, and  that $E_1\dotplus E_2\dotplus\cdots\dotplus E_s=\C^{n}$.  
	Thus, every vector $\varphi\in\C^{n}$ can be uniquely decomposed as $\varphi=\varphi_{(1)}+\varphi_{(2)}+\cdots+\varphi_{(s)}$, for unique $\varphi_{(j)}\in E_j$. The linear
	transformation $\varphi\mapsto \varphi_{(j)}$ is the \emph{projector} $P_j$ onto the
	subspace~$E_j$. Actually, the Hermitian matrix $\O$ is biunivocally determined by
	its eigenvalues and  projectors as $\O=\nu_1\cdot P_1 + \nu_2\cdot P_2 + \dots + \nu_s\cdot P_s$. 
	We recall that a matrix $P \in  \C^{n \times n}$ is a projector if and only if $P$ is Hermitian and idempotent, i.e. $P^2 = P$. 	 
	
	Let $\omega=\nradix{}$ be the $n$th root of the unity ($\omega^n=1$) and define the Vandermonde matrix $W\in\C^{n\times n}$ whose $(r,c)$th component is $\omega^{rc}$, for $0\leq r,c<n$. Let the $n\times n$ complex matrix
	$F_n=\frac{1}{\sqrt{n}}\cdot W.$
	It is easy to see that $F_n$ is the unitary matrix implementing the \emph{quantum Fourier transform}.
	Throughout the paper, it will be useful to recall that operating $F_n$ on the $j$th vector of the canonical basis yields the vector
	$\e_j\cdot F_n = \frac{1}{\sqrt{n}}\cdot (\omega^0,\omega^{(j-1)\cdot 1},\ldots,\omega^{(j-1)\cdot(n-1)})$. We remark that $|(\e_j\cdot F_n)_k|^2=\frac{1}{n}$, for every $1\le k\le n$.
	
	As we will see in the next section, in accordance with quantum mechanics principles (see, e.g., \cite{Hu92}), the state of a quantum finite state automaton at any given time during its computation is represented by a norm 1 vector from a suitable Hilbert space, the state evolution of the automaton is modeled by unitary matrices, and information on certain characteristics of the automaton are probabilistically extracted by measuring some ``observables'' represented by Hermitian matrices.
	
	\subsection{Quantum Finite State Automata}\label{lqfa}
	Here, we recall the model of a {\em Latvian quantum finite state automaton} \cite{Ambainis06} we are mostly interested~in. 
	We then quickly introduce {\em measure-once} quantum finite state automata \cite{brodsky2002characterizations,MC00} as a particular case of Latvian automata. Finally, we overview {\em measure-many}  quantum finite state automata \cite{kondacs1997power}.
\begin{definition}
Let $\Sigma$ be an input alphabet, $\r\notin\Sigma$ an endmarker symbol, and set $\Gamma=\Sigma\cup\{\r\}$.
A \emph{Latvian quantum finite automaton} {\em (}\lqfa\ {\em for short)} is a system
$
\A = ( Q, \Sigma, \pi_0, \{U(\sigma)\}_{\sigma\in\Gamma}, \{\O_{\sigma}\}_{\sigma\in\Gamma}, Q_{acc} )
$,~where
\begin{itemize}
\item $Q=\{q_1, q_2, \dots, q_n\}$ is the finite set of basis states; the elements of $Q$ span\footnote{We can associate with the set $Q=\{q_1, q_2, \dots, q_n\}$ of basis states the canonical basis $\set{\e_1,\ldots,\e_n}$ of the Hilbert space $\C^n$ (see Section \ref{lin})
where, for each $1\le i\le n$, we let $\e_i$ represent the basis state $q_i$. As the canonical basis spans $\C^n$, with a slight abuse of notation, we say that the elements of $Q$ spans $\C^n$.} the Hilbert space $\C^n$,
\item $Q_{acc}\subseteq Q$ is the set of accepting basis  states,
\item $\pi_0\in\C^n$ is the initial amplitude vector (superposition) satisfying $\|\pi_0\|=1$,
\item $U(\sigma)\in\C^{n\times n}$ is the unitary evolution matrix, for any $\sigma\in\Gamma$,
\item for any $\sigma\in\Sigma$, we let $\O_\sigma=\sum_{i=0}^{k_\sigma-1}c_i(\sigma)\cdot P_i(\sigma)$ be an observable (Hermitian matrix) on $\C^n$, where $\{c_0(\sigma),\ldots,c_{k_\sigma-1}(\sigma)\}$ is the set of all possible outcomes (eigenvalues) of measuring~$\O_\sigma$, and $\{P_0(\sigma),\ldots,P_{k_\sigma-1}(\sigma)\}$ are the projectors onto the corresponding eigenspaces,
\item we let $\O_\r=a\cdot P_{acc}(\r)+r\cdot (I^{(n)}-P_{acc}(\r))$ be the final observable, where $P_{acc}(\r)$ is the projector onto the subspace of $\C^n$ spanned by the states in $Q_{acc}$.
\end{itemize}
\end{definition}
\noindent Let us briefly describe the behavior of $\A$ on an input word  $w \r \in \Sigma^*\r$. 
	At any given time, the \emph{state} of~$\A$ is a {\em superposition of basis states} in $Q$ which is represented by a norm 1 vector $\xi\in\C^n$. We have that~$\xi_i\in\C$ is the \emph{amplitude} of the basis state $q_i$, while $|\xi_i|^2\in[0,1]$ is the \emph{probability} of observing $\A$ in the basis state~$q_i$.
	The computation of $\A$ on $w \r$ starts in the initial superposition $\pi_0$ by reading the first input symbol. Then, the transformations associated with each input symbol are applied in succession. 
	The transformation corresponding to a symbol $\sigma \in \Gamma$ consists of two steps:
	\begin{enumerate}
		\item {\em Evolution:} the matrix $U(\sigma)$ acts on the current state $\xi$ of $\A$, yielding the next state $\xi' = \xi \cdot U(\sigma)$.
		\item {\em Observation:} the observable $\O_\sigma$ is measured and the outcome $c_i(\sigma)$
		is seen with probability $\left\| \xi'\cdot P_i(\sigma) \right\|^2$; upon seeing $c_i(\sigma)$, according to 
		Copenhagen interpretation of quantum mechanics~\cite{Hu92}, the state of
		 $\A$ "collapses" to (norm 1) state $\xi'\cdot P_i(\sigma)/\left\| \xi'\cdot P_i(\sigma) \right\|$ and the computation continues, unless we are processing the endmarker~$\r$.
	\end{enumerate}
	Upon processing the endmarker $\r$, the final observable $\O_\r$ is measured yielding the probability of seeing~$\A$ in an accepting basis state. Therefore, the probability of accepting $w\in\Sigma^*$ is given by
	$$
	p_\A(w)=\sum_{i_1=0}^{k_{w_1}-1} \cdots \sum_{i_{|w|}=0}^{k_{w_{|w|}}-1}
	\left\| \pi_0\cdot U(w_1)\cdot P_{i_1}(w_1)\cdot\, \cdots\,\cdot U(w_{|w|})\cdot P_{i_{|w|}}(w_{|w|})\cdot U(\r)\cdot P_{acc}(\r) \right\|^2.
	$$
	The function $p_\A: \Sigma^* \rightarrow [0, 1]$ is the {\em stochastic event induced by $\A$}.
	The \emph{language recognized by~$\A$ with cut point $\lambda\in [0, 1]$} is the set of words $L_{\A,\lambda} =\sett{w\in\Sigma^*}{p_\A(w)>\lambda}.$ 
	The cut point $\lambda$ is said to be \emph{isolated} whenever there exists $\varrho\in \left(0,\frac{1}{2}\right]$ such that $|p_\A(w)-\lambda|\geq\varrho$, for any $w \in \Sigma^*$. The parameter $\rho$ is usually referred to as {\em radius of isolation}.
	
	 In general, a language
	$L\subseteq\Sigma^*$ is recognized with isolated cut point by a \lqfa\ whenever there exists a \lqfa\ $\A$ such that \  $(\inf\sett{p_\A(w)}{w\in L}-\sup\sett{p_\A(w)}{w\not\in L})>0$. In this case, we can compute the cut point as being $\lambda=\frac{1}{2}\cdot\left(\inf\sett{p_\A(w)}{w\in L}+\sup\sett{p_\A(w)}{w\not\in L}\right)$, with radius of isolation $\varrho= \frac{1}{2}\cdot(\inf\sett{p_\A(w)}{w\in L}-\sup\sett{p_\A(w)}{w\not\in L})$. Throughout the rest of the paper, for the sake of conciseness, we will sometimes be writing ``isolated cut point quantum finite automaton for a language'' instead of ``quantum finite automaton recognizing a language with isolated cut point''.
	Isolated cut point turns out to be one of the main language recognition policies within the literature of probabilistic devices.
	Its relevance in the realm of finite state automata is due to the fact that we can arbitrarily reduce the classification error probability of an input word $w$ by repeating a constant number of times (not depending on the length of $w$) its parsing and taking the majority of the answers. We refer the reader to ,e.g., \cite[Sec.~5]{Ra63}, where the notion of isolated cut point recognition is introduced and carefully analyzed.
	\smallskip
	
	One of the two original and most studied models of a quantum finite state automaton is the {\em measure-once} model  (\mo\ for short). An \mo\ can be seen as a particular \lqfa\ where, for any $\sigma\in\Sigma$, we have that $\O_\sigma=I^{(n)}$. Basically, this amounts to leave the computation of $\A$ undisturbed up to the final observation for acceptance. Thus, an \mo\ can be formally and more succinctly written as $\A = (Q,\Sigma,\pi_0, \set{U(\sigma)}_{\sigma \in \Gamma}, Q_{acc})$. The probability of $\A$ accepting the word $w\in \Sigma^*$ now simplifies as
	\smallskip
	
		$
		~~~~~~~~~~~~~~~~~~~~~~~~~~~~~~~~~~~~~p_\A(w) = 	\left\| \pi_0 \cdot U(w_1)\cdot\,\cdots\,\cdot U(w_{|w|})\cdot U(\r)\cdot P_{acc}(\r) \right\|^2.
		$
\smallskip
	
	Let us now switch to the other original model, namely a {\em measure-many quantum finite state automaton} (\mm\ for short). 
	Roughly speaking, an \mm\ $\A$ is defined as \lqfa\ but with the possibility of accepting/rejecting the input string \emph{before} reaching the endmarker. More precisely, the set~$Q$ of the basis states of $\A$ can be partitioned into {\em halting states}, which can be either {\em accepting} or {\em rejecting}, and {\em non halting states}, also called  {\em go} states, i.e., $Q=Q_{acc}\cup Q_{rej}\cup Q_{go}$. Following such a state partition, the sole observable 
	$\O = a \cdot P_{acc} + r \cdot P_{rej} + g \cdot P_{go}$, whose projectors map onto the subspaces spanned by the corresponding basis states, is associated with \emph{each} symbol in $\Gamma$. At each step, the observable~$\O$ is measured and the computation of $\A$ continues (unless we are processing $\r$) only if the outcome $g$ is seen. Instead, if the outcome $a$ ($r$) is seen, then $\A$ halts and accepts (rejects). Formally, the \mm\ $\A$ can be written as  $\A = (Q,\Sigma,\pi_0, \{U(\sigma)\}_{\sigma\in \Gamma}, \O, Q_{acc})$, and the probability of accepting the word $w\r=w_1\cdots w_nw_{n+1}$ is
	\smallskip
	
	$
	~~~~~~~~~~~~~~~~~~~~~~~~~~~~~~~~~~~p_\A(w) = 	\sum_{k=1}^{n+1} \| \pi_0 \cdot ( \prod_{i=1}^{k-1} U(w_i)\cdot P_{go})\cdot  U(w_k)\cdot P_{acc} \|^2.
	$
	\smallskip
		
	It is well known (see, e.g., \cite{brodsky2002characterizations,MC00}) that the class of languages recognized by isolated cut point \mos\ coincides with the class of group languages. Notice that finite languages are not group languages, and hence they cannot be accepted by isolated cut point \mos. Isolated cut point
	 \lqfas\ are proved in \cite{Ambainis06,Mercer07} to be strictly more powerful than isolated cut point \mos, since their recognition power coincides with the class of block group languages. An equivalent characterization states that a language is recognized by an isolated cut point \lqfa\ if and only if it belongs to the boolean closure of languages of the form $L_1a_1L_2a_2\cdots a_kL_{k+1}$, for $a_i\in\Sigma$, group language $L_i\subseteq\Sigma^*$, and $|\Sigma|>1$.
Finally, the recognition power of isolated cut point \mms\ still remains an open question. However, it is know that \mms\ are strictly more powerful than \lqfas\ but strictly less powerful than \dfas. 
In fact, isolated cut point \mms\  can recognize the language $a\Sigma^*$ which cannot be accepted by isolated cut point \lqfas~\cite{Ambainis06}. On the other hand, isolated cut point \mms\  cannot recognize the language $\Sigma^*a$, for $|\Sigma|>1$ and $a\in\Sigma$~\cite{kondacs1997power}.
	
	\section{Isolated Cut Point \lqfas\ for Words Longer than $T$}\label{sec:Mh_description}

Here, we design an isolated cut point \lqfa\ recognizing the unary language $\sigma^{\ge T}$,
for any given~\mbox{$T>0$} (i.e., the set of unary strings whose length is greater than or equal to $T$). As it will be clear in the next section, this \lqfa\ will be a relevant component in the modular construction of isolated cut point \lqfas\ for unary regular languages.

Our design pattern is inspired by \cite{Ambainis06,Mercer07}, where the authors provide an isolated cut point \lqfa\ for the language $\Sigma^* a_1 \Sigma^*a_2  \cdots\,  a_k \Sigma^*$, with $a_i\in\Sigma$ and $|\Sigma| >1$.
So, we focus on recognizing the unary version of $\Sigma^* a_1 \Sigma^*a_2  \cdots\,  a_T \Sigma^*$ yielded by fixing $a_1 = \dots =  a_T =\sigma$ and $\Sigma=\set{\sigma}$, namely, the desired language $a_1\cdots a_T\Sigma^*=\!\sigma^{\ge T}$.
	We adapt the construction in \cite{Ambainis06,Mercer07}, and inductively exhibit a family~$\set{M^{(\ell)}}_{\ell\ge1}$
	of~\lqfas\ such that: {\em (i)} $M^{(\ell)}$ recognizes the language $\sigma^{\ge \ell}$ with isolated cut point, and {\em (ii)} $M^{(\ell)}$ is constructed by ``expanding'' $M^{(\ell-1)}$. So,
the desired isolated cut point \lqfa\ for $\sigma^{\ge T}$ will result after $T$ ``expansions'', starting from the \lqfa~$M^{(1)}$. We provide a detailed analysis of the stochastic behavior of~$M^{(\ell)}$ machines, emphasizing cut points, isolations and their size (i.e., number of their basis states).
	In this section, to have a convenient notation, we will be using $A_\sigma$ for the evolution operator of our \lqfas.
\medbreak

\noindent {\bf Base of the construction:} 
	For the induction base, we define the \lqfa\ $M^{(1)}$ for the language $ \sigma^{\ge1}$ as
	\smallskip
	
	$
	~~~~~~~~~~~~~~~~~~~~~~~~~~~~~~~~M^{(1)}=(Q^{(1)}, \set{\sigma}, \pi_0, \{A^{(1)}_\sigma, A^{(1)}_\r \}, \{\O_\sigma^{(1)}, \O_\r^{(1)}\}, Q_{acc}^{(1)}),
	$
	\smallskip
	
	\noindent where
	$Q^{(1)}= \{q_0, \dots, q_{n-1}\}$ is the set of $n$ basis states,
	$\pi_0=\e_1$  meaning that $M^{(1)}$ starts in the state $q_0$, 
	$Q_{acc}^{(1)} = Q^{(1)} \backslash \set{q_0}$ is the set of $n-1$ accepting states. For the evolution matrices, we let $A^{(1)}_\sigma= F_n$ (the quantum Fourier transform) and $A^{(1)}_{\r}= I$ (the identity matrix). The observable $\O_\sigma^{(1)}$ is the {\em canonical observable} defined by the projectors  $ \{ {\e_1}^\dag \cdot \e_1,\; {\e_2}^\dag \cdot \e_2,\; \dots, {\e_n}^\dag \cdot \e_n\}$. 	
	By measuring $\O_\sigma^{(1)}$ on $M^{(1)}$ being in the superposition $\xi\in\C^n$, we will see $M^{(1)}$
	in the basis state $q_{i-1}$ with probability $\|\xi\cdot({\e_i}^\dag \cdot \e_i) \|^2=|\xi_i|^2$. Upon such an outcome,
	the state of $M^{(1)}$ clearly collapses to $\e_i$. The final observation $\O_\r^{(1)}$ projects onto the subspace spanned by the accepting basis states $\set{q_1,\ldots,q_{n-1}}$.

	The automaton $M^{(1)}$ behaves as follows: when the first input symbol is read, the state of $M^{(1)}$ becomes
	$\pi_0\cdot A^{(1)}_\sigma=\e_1\cdot F_n$, upon which the canonical observation is measured. As noticed at the end of Section \ref{lin}, such a measurement will cause $M^{(1)}$ to move from $q_0$ to some basis state $q_i$, with $0\le i \le n-1$, uniformly at random (i.e., with probability $|(\e_1\cdot F_n)_{i+1}|^2=\frac{1}{n}$). After processing (again, by quantum Fourier transform followed by measuring the canonical  observable) the next input symbol from being in the state $\e_i$, we again find $M^{(1)}$ in a basis state uniformly at random. Such a dynamics continues unaltered, until the endmarker is reached and processed by the identity matrix. At this point, the final observation~$\O_\r$ is measured, and an accepting state is easily seen to be reached with probability $|Q_{acc}^{(1)}| \cdot \frac{1}{n}=(\frac{n-1}{n})$. Clearly, processing the empty string leaves $M^{(1)}$ in the non accepting state $q_0$ with certainty. Therefore, $p_{M^{(1)}}(\varepsilon)=0$, while for $k>0$ we have $p_{M^{(1)}}(\sigma^k)=(\frac{n-1}{n})$. So, $M^{(1)}$ recognizes the language $\sigma^{\ge 1}$ with isolated cut point.
\medbreak

\noindent {\bf Inductive step of the construction:}
	For the inductive step, we show how to build the isolated cut point \lqfa\ $M^{(\ell)}$ for the language $\sigma^{\ge \ell}$ from the \lqfa\ $M^{(\ell-1)}$ for the language $\sigma^{\ge \ell-1}$, this latter \lqfa\ being given by inductive hypothesis. We define
\smallskip

$~~~~~~~~~~~~~~~~~~~~~~~~~~~~~~~~M^{(\ell)}=(Q^{(\ell)}, \set{\sigma}, \pi_0, \{A^{(\ell)}_\sigma, A^{(\ell)}_\r \}, \{\O_\sigma^{(\ell)}, \O_\r^{(\ell)}\}, Q_{acc}^{(\ell)}),$
\smallskip

\noindent where the set $Q^{(\ell)}$ of basis states consists of the previous set $Q^{(\ell-1)}$ of basis states,
	plus $(n-1)$ new basis states per each state in $Q_{acc}^{(\ell-1)}$.
		 We let $Q_{acc}^{(\ell)}$ be the set containing these $(n-1)\cdot|Q_{acc}^{(\ell-1)}|$ new states, with 
		 $|Q_{acc}^{(\ell)}|=(n-1)^{\ell}$. Therefore,
 $Q^{(\ell)}= Q^{(\ell-1)} \cup Q_{acc}^{(\ell)}=\{q_0\} \cup Q_{acc}^{(1)} \cup Q_{acc}^{(2)} \cup \dots  \cup Q_{acc}^{(\ell)}$ with $|Q_{acc}^{(i)}|=(n-1)^i$,
so that  $|Q^{(\ell)}|=\sum_{i=0}^\ell (n-1)^i=\frac{(n-1)^{(\ell+1)}-1}{n-2}$.	The initial superposition is $\pi_0=\e_1$.
We let $A_\r^{(\ell)}=I$ and $A^{(\ell)}_\sigma = B^{(\ell)}\cdot \tilde{A}^{(\ell-1)}$, where $\tilde{A}^{(\ell-1)}$ is the transformation acting as $A^{(\ell-1)}_\sigma$ on $Q^{(\ell-1)}\subset Q^{(\ell)}$, and as the identity elsewhere. Instead,
		$B^{(\ell)}$ is an additional operator working as follows. 
		For any~$\tilde{q} \in Q^{(\ell-1)}_{acc}$, let $Q_{\tilde{q}}= \{ \tilde{q}_1, \dots ,\tilde{q}_{n-1}\} \subset Q^{(\ell)}_{acc}$ be the set of the $n-1$ new added accepting states associated with~$\tilde q$. 
		Thus, the operator $B^{(\ell)}$ first acts as $F_n$ on $\set{\tilde{q}} \cup Q_{\tilde{q}}$ for every
		$\tilde q\in Q_{acc}^{(\ell-1)}$, then it measures $\O_\sigma^{(\ell)}$ being the canonical observable on $Q_{acc}^{(\ell-1)}\cup Q_{acc}^{(\ell)}$ plus the identity projector on the remaining basis states. The final observable $\O_\r^{(\ell)}$ as usual projects onto the subspace spanned by $Q_{acc}^{(\ell)}$.
Actually, the automaton so far constructed does not perfectly comply with the definition of a \lqfa\ given in Section \ref{lqfa} since $A_\sigma^{(\ell)}$ is not a unitary matrix. However, \mbox{\cite[Claim~1]{Ambainis06}}
		ensures that the action of the operator $B^{(\ell)}\cdot \tilde{A}^{(\ell-1)}$ followed by measuring $\tilde \O_\sigma^{(\ell-1)}$ (the observable of~$M^{(\ell-1)}$ extended to  $Q^{(\ell)}$ by the identity projector onto $Q^{(\ell)}_{acc}$) can be expressed as a unitary matrix followed by measuring a suitable observable. This last detail possibly enlarges the dimension of the Hilbert space for~$M^{(\ell)}$ by a factor bounded by $n^\ell$. The stochastic event induced by $M^{(\ell)}$ will be discussed later.
		\smallskip

	To clarify the architecture and behavior of this family of automata, we now describe the \lqfa~$M^{(3)}$ recognizing the language $\sigma^{\ge3}$ with isolated cut point. We have 
	\smallskip
	
	$~~~~~~~~~~~~~~~~~~~~~~~~~~~~~~~~M^{(3)} = (Q^{(3)}, \set{\sigma}, \pi_0, \{A^{(3)}_\sigma, A^{(3)}_\r \}, \set{\O_\sigma^{(3)}, \O_\r^{(3)}}, Q^{(3)}_{acc}),$
	\smallskip
	
	\noindent where we let the set of basis states be 
	$Q^{(3)} = \{q_0\} \cup Q_{acc}^{(1 )}\cup Q_{acc}^{(2)} \cup Q_{acc}^{(3)}$ with
	$Q_{acc}^{(1 )}=\sett{q_i}{1\le i\le n-1}$, 	$Q_{acc}^{(2 )}=\sett{q_{i,j}}{1\le i,j\le n-1}$, and
	$Q_{acc}^{(3 )}=\sett{q_{i,j,k}}{1\le i,j,k\le n-1}$. We remark that $ Q_{acc}^{(3)}$ is the set of $(n-1)^3$ accepting basis states of $M^{(3)}$. 
We can regard basis states as partitioned into three groups reflected by the number of subscripts attributed to each basis state; 
each group of states is added in a subsequent step of the inductive construction.
	The general structure of the state (superposition) of $M^{(3)}$ is a norm~1 vector in $\C^{|Q^{(3)}|}$ of the following form, with $\alpha(q)$ denoting the amplitude of the basis state $q$: 
	{\small \begin{align*}
		[\alpha(q_0), \alpha(q_1), \alpha(q_{1,1}), [\ldots\, \alpha(q_{1,1,k})\, \ldots], \alpha(q_{1,2}), [\ldots\,\alpha(q_{1,2,k})\, \ldots], \; \dots \; , \alpha(q_{1,n-1}), [\ldots\,\alpha(q_{1,n-1,k}) \,\ldots],\\
		\; \alpha(q_2), \alpha(q_{2,1}), [\ldots\, \alpha(q_{2,1,k})\, \ldots], \alpha(q_{2,2}), [\ldots\,\alpha(q_{2,2,k})\, \ldots], \; \dots \; , \alpha(q_{2,n-1}), [\ldots\,\alpha(q_{2,n-1,k}) \,\ldots],\\
		\vdots~~~~~~~~~~~~~~~~~~~~~~~~~
		\\
		\; \alpha(q_{n-1}), \alpha(q_{n-1,1}), [\ldots\, \alpha(q_{n-1,1,k})\, \ldots], \alpha(q_{n-1,2}), [\ldots\,\alpha(q_{n-1,2,k})\, \ldots], \; \dots \; , \alpha(q_{n-1,n-1}), [\ldots\,\alpha(q_{n-1,n-1,k}) \,\ldots]].\\[.2cm]
\mbox{\normalsize (*)\ \ Form of states (superpositions) of $M^{(3)}$.~~~~~~~~~~~~~~~~~~~~~~~~~~~~~~~~~~~~~~~~~~~~~}
	\end{align*}
	}
	\vspace{-.6cm}
	
\noindent As usual, we let $\pi_0=\e_1$. The evolution matrices of $M^{(3)}$  are $A^{(3)}_\r=I$, while we have $A_\sigma^{(3)} = B^{(3)}\cdot \tilde{B}^{(2)}\cdot \tilde{A}^{(1)}$, where each matrix in the product acts on  levels of the basis states as follows: 
	$\tilde{A}^{(1)}$ affects the states in $\set{q_0}\cup Q_{acc}^{(1)}$,\ \ $\tilde{B}^{(2)}$ the states in $Q^{(1)}_{acc}\cup Q^{(2)}_{acc}$, and $B^{(3)}$ the states in $Q^{(2)}_{acc}\cup Q^{(3)}_{acc}$.
	From now on, it will be useful to describe the dynamic of $M^{(3)}$ by displaying the sequence of the stochastic vectors
	obtained by squaring the amplitudes in the superpositions of the form in (*). In such vectors, the value $|\alpha(q)|^2$ of the component associated with $q$ represents
	 the probability for $M^{(3)}$ of being in the basis state~$q$. This stochastic dynamic description turns out to be appropriate as $M^{(3)}$ uses the canonical observable after each quantum Fourier transform operation.  
Upon reading a symbol $\sigma$, the \lqfa\ $M^{(3)}$ executes~$A^{(3)}_\sigma$ followed by measuring~$\tilde\O^{(1)}_\sigma$: formally, we write $A^{(3)}_\sigma\downarrow \tilde\O^{(1)}_\sigma$. 
	This operation distributes the probability differently in the three group of basis states $Q^{(1)}$, $ Q_{acc}^{(2)}$ and $Q_{acc}^{(3)}$.
	In particular, the probability values turn out to be identical within each group of basis states, for each step of computation (except for the initial superposition $\pi_0$). 
	Therefore, the form of the stochastic vector at each step of computation is
	\begin{flalign*}
		[x,\, x,\, y,\, [\cdots\, z\, \cdots],\, y,\, [\cdots\, z\, \cdots], \; \dots \; ,\, y,\, [\cdots\, z\, \cdots],\\
		\; x,\, y,\, [\cdots\, z\, \cdots],\, y,\, [\cdots\, z\, \cdots], \; \dots \; , y,\, [\cdots\, z\, \cdots],\\
	  \vdots~~~~~~~~~~~~~~~~~~~~~~~~~
	  \\
		\; x,\, y,\, [\cdots\, z\, \cdots],\, y,\, [\cdots\, z\, \cdots], \; \dots \; ,\, y,\, [\cdots\, z\, \cdots]],
	\end{flalign*}
	where $x$ is the probability value for the states in $Q^{(1)}$,\ \ $y$ for the states in $ Q_{acc}^{(2)} $, and $z$ for the (accepting) states in $Q_{acc}^{(3)}$. Thus, the current accepting probability is $(n-1)^3\cdot z$.

Now, let $x(k)$,\  $y(k)$, and $z(k)$ be the above basis states probabilities after processing the $k$th input symbol. 
We are going to establish the dependence of such values from $x(k-1)$,\ $y(k-1)$, and $z(k-1)$ in order
to single out a closed formula for the stochastic event $p_{M^{(3)}}$. To this aim, for reader's ease of mind,
 a graphical representation is given in Figure~\ref{fig:M3_trees}, of how one step of the evolution-plus-observation $A^{(3)}_\sigma\downarrow \tilde\O^{(1)}_\sigma$ affects the probability values in each different group of basis states.  
\begin{figure}[hbt]
\begin{center}
\includegraphics[scale=.36]{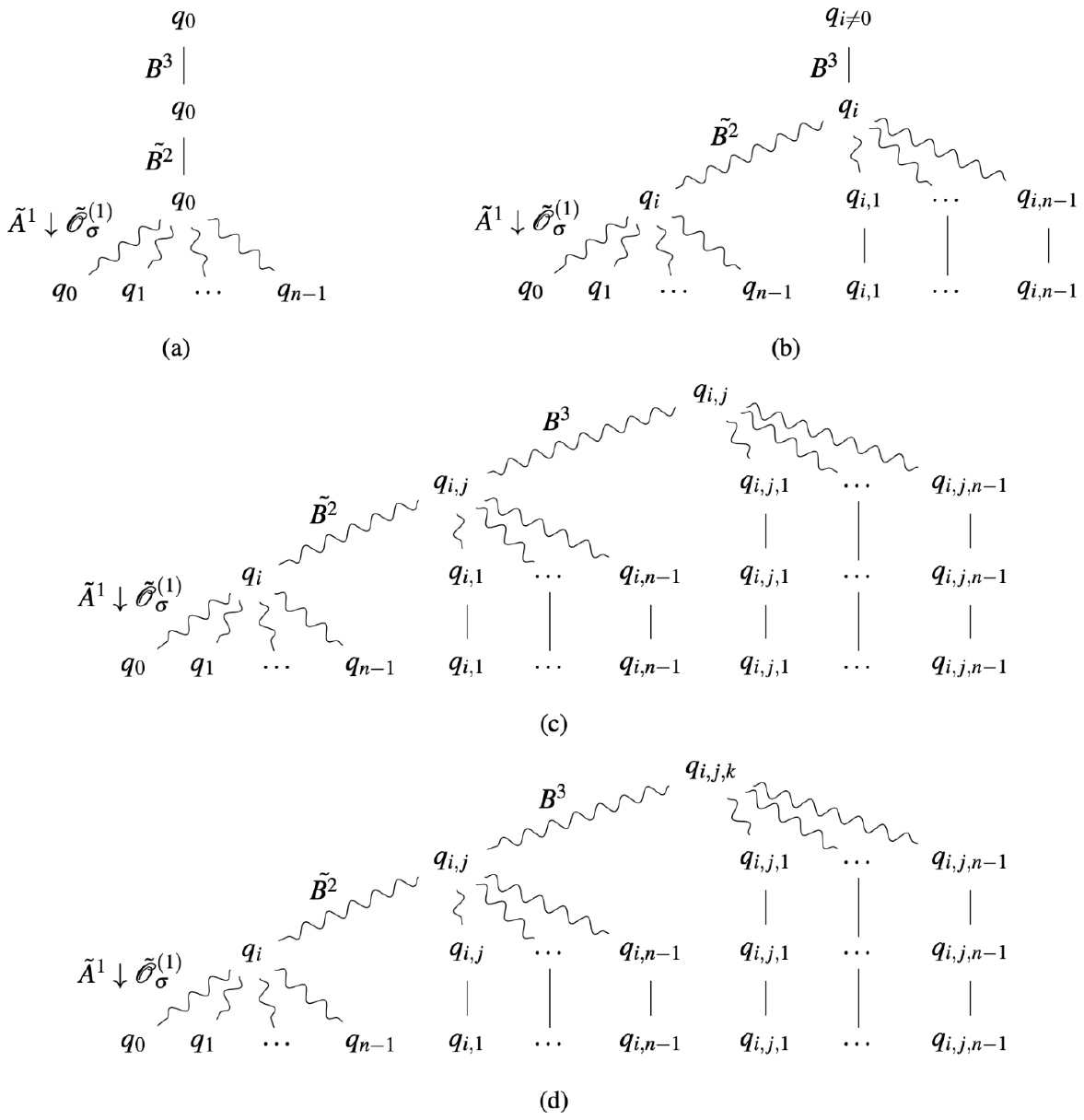}
\vspace{.1cm}
\caption{Stochastic representation of a computation step of $M^{(3)}$ on the symbol $\sigma$ for basis states of different groups.  The notation $\tilde{A}^1\downarrow \tilde\O^{(1)}_\sigma$ means that $\tilde{A}^1$ is applied and then the observable
$\tilde\O^{(1)}_\sigma$ is~measured. Wave (straight) edges indicate basis state transitions occurring with probability $\frac{1}{n}$ (with certainty). For instance, the tree in (b) says that, starting from $q_{i\ne0}$ for a fixed $i$ and after one step of computation,  we will observe $M^{(3)}$ in $q_0$ with probability $\frac{1}{n^2}$. Note that there exist $n-1$ trees of the form (b) leading to~$q_0$.}\label{fig:M3_trees}
\end{center}
\end{figure}

	\newpage
	
	\noindent Let us focus, e.g., on~$x(k)$. 
The probability $x(k)$ depends on $x(k-1)$,\ \ $y(k-1)$, and $z(k-1)$ as follows:
	\begin{itemize}
	\item Figure~\ref{fig:M3_trees}(a) shows that the basis state $q_0$ contributes with $\frac{1}{n}\cdot x(k-1)$.
	\item Figure~\ref{fig:M3_trees}(b) shows the contribution of each basis states in $Q^{(1)}_{acc}$,  which is $\frac{1}{n^2}\cdot x(k-1)$; given that $|Q^{(1)}_{acc}|=(n-1)$, the total contribution is $\frac{(n-1)}{n^2}\cdot x(k-1)$.
	\item Figure~\ref{fig:M3_trees}(c) shows that the total contribution given by $y(k-1)$ elements (i.e., by the $(n-1)^2$  basis states in $Q^{(2)}_{acc}$) is $\frac{(n-1)^2}{n^3}\cdot y(k-1)$.
	\item Figure~\ref{fig:M3_trees}(d) shows that the total contribution given by $z(k-1)$ elements (i.e., by the $(n-1)^3$  basis states in $Q^{(3)}_{acc}$) is $\frac{(n-1)^3}{n^3}\cdot z(k-1)$.
	\end{itemize}
By analogous reasonings, we can obtain recurrences for $y(k)$ and $z(k)$, globally yielding the system 
	\begin{align} \label{eq:recursion_sys}
		\begin{cases}
			x(k) = \frac{1}{n}\cdot x(k - 1) + \frac{(n-1)}{n^2}\cdot x(k - 1) + \frac{(n-1)^2}{n^3}\cdot y(k - 1) +\frac{(n-1)^3}{n^3}\cdot z(k - 1)\\
			y(k)=  \frac{1}{n}\cdot x(k - 1) + \frac{(n-1)}{n^2}\cdot y(k - 1) +\frac{(n-1)^2}{n^2}\cdot z(k - 1)\\
			z(k) = \frac{1}{n}\cdot y(k - 1) + \frac{(n-1)}{n}\cdot z(k - 1).
		\end{cases}
	\end{align} 
	The base for this system of recurrences is the probability distribution after reading the first symbol $\sigma$,~i.e.:
	\begin{align} \label{eq:recursion_base}
	\begin{cases}
		x(1) = \frac{1}{n}\\
		y(1)= z(1) = 0.
	\end{cases}
	\end{align} 
	From the system (\ref{eq:recursion_sys}), the reader may verify that at each computation step the probability ``shifts'' towards the next deeper level of the basis states until reaching the basis states in $Q^{(3)}_{acc}$. In fact,
	after the first~step (yielding probabilities in (\ref{eq:recursion_base})), only the $x$-components have non null values.
	After the second step, only the $x$- and $y$-components have values different from $0$, while the value of the $z$-components is still~$0$. This shows that $M^{(3)}$ rejects with certainty the strings in $\sigma^{\le2}$. After the third step, all the components have non null values; in particular, $z(3)=\frac{1}{n^3}$, so that the accepting probability of the string $\sigma^3$~attains $|Q_{acc}^{(3)}|\cdot z(3)=(\frac{n-1}{n})^3$. 
	By solving the system (\ref{eq:recursion_sys}),
	we get a closed formula for $z(k)$, with $k\ge2$, as
$$		z(k) 
		= \frac{1}{n(n-1)^2}\cdot \left( 1 -
		\frac{(\frac{2n-2}{n^2})^{k-2}\cdot (n-1)^2 + 1 }{(n-1)^2+1}
		\right).
$$
	This allows us to evaluate the accepting probability of $M^{(3)}$ for any string in $\sigma^{*}$ as
	\begin{equation}\label{eq: m3_acc_prob}
	p_{M^{(3)}}(\sigma^k) =  |Q^{(3)}_{acc}|\cdot z(k) = \begin{cases}
		0 & \mbox{if $k\le 2$}\\
		\frac{n-1}{n}\cdot \left( 1 -
		\frac{(\frac{2n-2}{n^2})^{k-2}\cdot (n-1)^2 + 1 }{(n-1)^2+1} 
		\right)& \mbox{if $k\ge 3$}.
		\end{cases}
	\end{equation} 
Equation (\ref{eq: m3_acc_prob}) shows that $M^{(3)}$ recognizes $\sigma^{\ge3}$ with isolated cut point.
Clearly, the stochastic event induced by $M^{(3)}$ depends on the number $n$ of the basis states of $M^{(1)}$,  
the initial automaton of the inductive construction. Figure \ref{fig:graph_m3} displays $p_{M^{(3)}}$ for some values of $n$. As expected, the higher $n$ grows, the better the isolation around the cut point becomes. 
\newpage

\begin{center}
\begin{figure}[hbt]
~~~~~~\includegraphics[scale=.4]{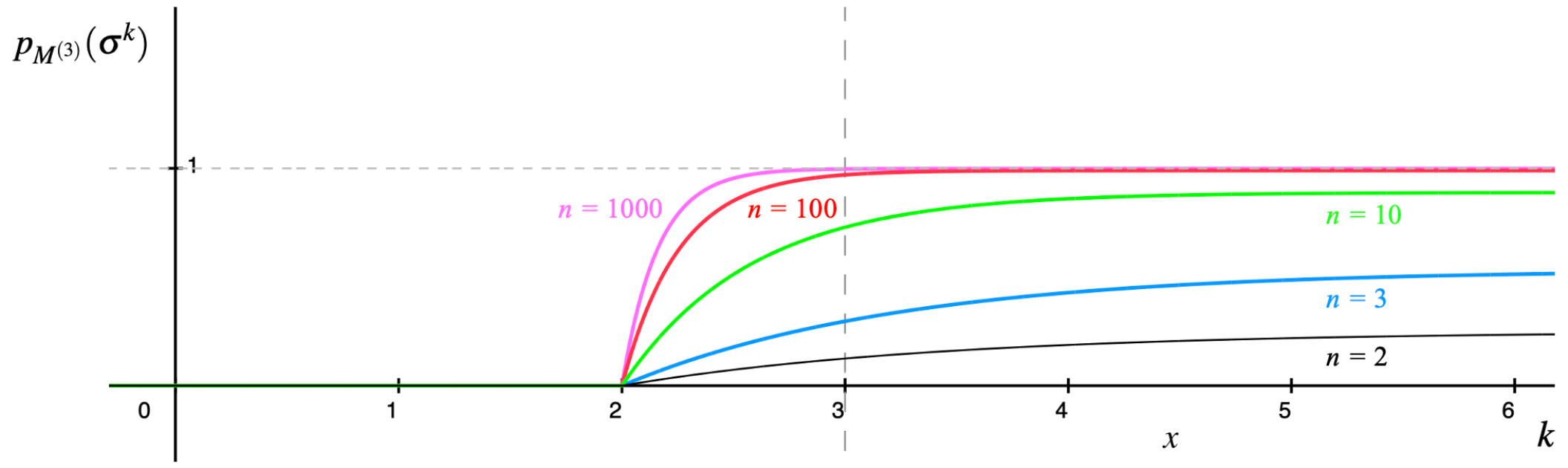}
		\caption{The (``continuous version'' of the) stochastic events induced by $M^{(3)}$ according to
		Equation~\ref{eq: m3_acc_prob}, for different values of the number $n$ of  basis states of $M^{(1)}$, the base module inductively leading~to~$M^{(3)}$.}
		\label{fig:graph_m3}\end{figure}
\end{center}
~\\
\vspace{-1.8cm}

\noindent 	Now, we consider the general \lqfa\ $M^{(\ell)}$, and derive the system of recurrences for its stochastic dynamic.
	The set of basis states of $M^{(\ell)}$ is now partitioned into $\ell$ groups. For $1\le h\le \ell$, we denote by 
	$x_h(k)$ the probability for $M^{(\ell)}$ of
	being in a basis state of the $h$th group, after processing $k$ input symbols.
	  The system of recurrences for $M^{(\ell)}$ generalizes the system ({\ref{eq:recursion_sys}}) as follows:
{
	\begin{equation}\label{eq:rec_general}
	\begin{cases}
		x_1(k) = \frac{1}{n}\cdot x_1(k - 1) + \sum_{j=1}^{\ell-1} \frac{(n-1)^j}{n^{j+1}}\cdot x_j (k - 1)  + \frac{(n-1)^\ell}{n^\ell}\cdot x_\ell(k - 1) \\	
		x_2(k)= \left(\sum_{j=0}^{\ell-2} \frac{(n-1)^j}{n^{j+1}} \cdot x_{j+1}(k - 1) \right)  + \frac{(n-1)^{\ell-1}}{n^{\ell-1}}\cdot x_\ell(k - 1) \\
		~~~~~~~~~~~~~~~~~~~\vdots\\
		x_h(k)=  \left(\sum_{j=0}^{\ell-h} \frac{(n-1)^j}{n^{j+1}}\cdot x_{j+h-1}(k - 1)\right) + \frac{(n-1)^{\ell-h+1}}{n^{\ell-h+1}}\cdot x_{\ell}(k - 1) \\
		~~~~~~~~~~~~~~~~~~~\vdots\\
		x_\ell(k) = \frac{1}{n}\cdot x_{\ell-1}(k - 1)  + \frac{(n-1)}{n}\cdot x_{\ell}(k - 1),
	\end{cases}
	\end{equation}
	}
	
	\noindent with initial values $x_1(1)=\frac{1}{n}$, and $x_h(1)=0$ for every $2\le h\le\ell$.
	We show the validity of this system of recurrences by induction, having, e.g., the system  (\ref{eq:recursion_sys}) for the automaton $M^{(3)}$ as base case. 
	By inductive hypothesis, we assume the system of recurrences for $M^{(\ell-1)}$, and we build the system~(\ref{eq:rec_general}) for~$M^{(\ell)}$. 
	We consider the set of trees representing one step of the computation of our automata, starting from basis states of different groups. E.g., Figure \ref{fig:M3_trees} displays the four different types of trees for $M^{(3)}$, one per each group of basis states, plus one for the evolution from the state $q_0$.  So, for $M^{(\ell)}$ we are going to provide $\ell$ of such trees, plus the one for $q_0$. Let us explain how to obtain them from the trees of $M^{(\ell-1)}$. Let $T^{(\ell-1)}_j$ be a tree representing one step of the evolution of $M^{(\ell-1)}$ on a basis state of group $1\le j\le \ell-1$, namely, a basis state from $Q^{(j)}_{acc}$. Moreover, let $T^{(\ell-1)}_0$ be the tree for $q_0$. The evolution for $M^{(\ell)}$  is\ \ $A_\sigma^{(\ell)} = B^{(\ell)}\cdot \tilde{A}^{(\ell-1)}$.
	\begin{figure}
		\centering
		\begin{subfigure}[t]{0.45\textwidth}
			\centering
			\resizebox{\linewidth}{!}{
			\begin{tikzpicture} 
				[every tree node/.style={ellipse},
				level distance=1.25cm,sibling distance=0.1cm,
				edge from parent path={(\tikzparentnode) -- (\tikzchildnode)},
				decoration=snake]
				]
				\Tree
				[.$q_{i_1,\dots,i_{\ell-1}}$
				\edge[decorate] node[auto=right] {$B^\ell$};
				[\edge; \node[draw, blue, regular polygon, regular polygon sides = 3, inner sep = 1pt, minimum size=100pt, yshift=-15pt]{$T_{\ell-1}^{\ell-1}$}; ]
				\edge [decorate];
				[.$q_{i_1,\dots,i_{\ell-1}, 1}$
				[.$\dots$ 
				[.$q_{i_1,\dots,i_{\ell-1}, 1}$ ] 
				]
				]
				\edge [decorate];
				[.$\dots$
				[.$\dots$ 
				[.$\dots$ ] 
				]
				]
				\edge [decorate];
				[.$q_{i_1,\dots,i_{\ell-1}, n-1}$
				[.$\dots$ 
				[.$q_{i_1,\dots,i_{\ell-1}, n-1}$ ] 
				]
				]
				]
			\end{tikzpicture}
		}
			\caption{}
			\label{fig:mltree_l}
		\end{subfigure}
		\hfill
		\begin{subfigure}[t]{0.45\textwidth}
			\centering
			\resizebox{\linewidth}{!}{
			\begin{tikzpicture}
				[every tree node/.style={ellipse},
				level distance=1.25cm,sibling distance=0.1cm,
				edge from parent path={(\tikzparentnode) -- (\tikzchildnode)},
				decoration=snake]
				\Tree
				[.$q_{i_1,\dots,i_{\ell}}$
				\edge [decorate] node[auto=right] {$B^\ell$};
				[\edge; \node[draw, blue, regular polygon, regular polygon sides = 3, inner sep = 1pt, minimum size=100pt, yshift=-15pt]{$T_{\ell-1}^{\ell-1}$}; ]
				\edge [decorate];
				[.$q_{i_1,\dots,i_{\ell-1}, 1}$
				[.$\dots$ 
				[.$q_{i_1,\dots,i_{\ell-1}, 1}$ ] 
				]
				]
				\edge [decorate];
				[.$\dots$
				[.$\dots$ 
				[.$\dots$ ] 
				]
				]
				\edge [decorate];
				[.$q_{i_1,\dots,i_{\ell-1}, n-1}$
				[.$\dots$ 
				[.$q_{i_1,\dots,i_{\ell-1}, n-1}$ ] 
				]
				]
				]
			\end{tikzpicture}
			}
			\caption{}
			\label{fig:mltree_l+1}
		\end{subfigure}
		\caption{The form of the tree $T^\ell_{\ell-1}$ in (a), and of the tree $T^\ell_{\ell}$ in (b) for the automaton
		 $M^{(\ell)}$. Within both these two trees, the tree $T_{\ell-1}^{\ell-1}$ turns out to be a sub-tree.}
		\label{fig:mltree_induction}
	\end{figure}
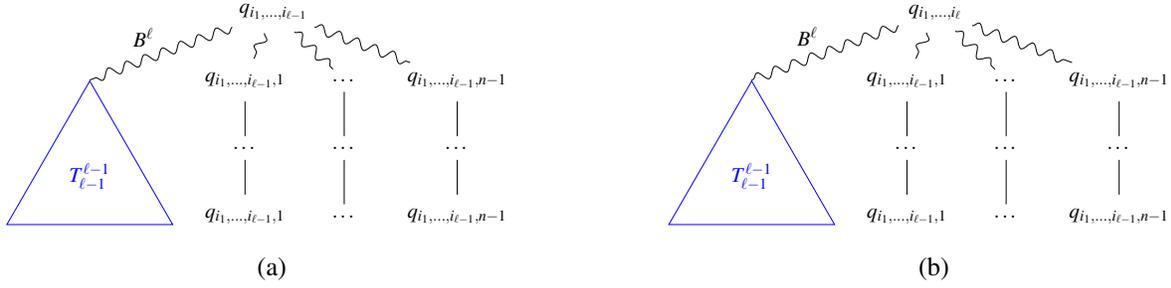
Thus, the behavior of $M^{(\ell)}$ is described by $\ell +1$ trees with the following structure:
	\begin{itemize}
		\item The trees $T^{(\ell)}_j$ for $0\le j< \ell-1$ are basically the trees $T^{(\ell-1)}_j$ with a preliminary 
		step due to the action of $B^{(\ell)}$. Since in these trees the root is labeled by a basis state of level $j<\ell-1$, such a preliminary step coincides with the identity evolution.
		\item Even the trees $T^{(\ell)}_{\ell -1}$ and $T^{(\ell)}_{\ell}$ have the action of $B^{(\ell)}$ 
		as a preliminary step. However, in these~cases, $B^{(\ell)}$ acts as $F_n$ on the basis states of groups $\ell-1$ in the tree $T^{(\ell)}_{\ell -1}$, and $\ell$ in the tree $T^{(\ell)}_{\ell}$.  The structure of these two trees, both containing the tree $T_{\ell-1}^{\ell-1}$ as a sub-tree, is presented in Figure \ref{fig:mltree_induction}.
	\end{itemize}
	\newpage
	~\\
	\vspace{-1.1cm}
	
\noindent 	It is now possible to properly justify the system (\ref{eq:rec_general}) by using the induction step.
	Starting from the system of recurrences for $M^{(\ell-1)}$, we show how it modifies towards the system for 
$M^{(\ell)}$. Clearly, a new recurrence for $x_\ell(k)$ (i.e., the probabilities for basis states of group $\ell$, the accepting states for $M^{(\ell)}$) is added at the end of the system. This component receives contributions only from the trees~$T^{(\ell)}_{\ell -1}$~and~$T^{(\ell)}_{\ell}$ weighted, respectively, by $x_{\ell-1}(k-1)$ and $x_{\ell}(k-1)$. Precisely, from the former tree we get the contribution $\frac{1}{n}\cdot x_{\ell-1}(k-1)$, from the latter ($n-1$ different trees) the contribution is $\frac{n-1}{n}\cdot x_{\ell}(k-1)$. For $x_h(k)$, with $1\le h\le\ell-1$, we note that
the only modified contribution is the one carried by $x_{\ell-1}(k-1)$; moreover a new contribution from 
$x_{\ell}(k-1)$ is added. Even in this case, the trees $T^{(\ell)}_{\ell -1}$ and $T^{(\ell)}_{\ell}$ account for
these modifications: the new coefficient of $x_{\ell-1}(k-1)$ is the old one for $x_{\ell-1}(k-1)$ (i.e., the one associated with $x_{\ell-1}(k-1)$ in the system for $M^{(\ell-1)}$)  multiplied by~$\frac{1}{n}$, while the coefficient of the new contribution $x_{\ell}(k-1)$ is the old one for $x_{\ell-1}(k-1)$ multiplied by
$\frac{(n-1)}{n}$.
\smallskip

By simply applying repeated substitutions in the system (\ref{eq:rec_general}), one may verify that,
 for $1\le k\leq \ell$, the value $x_k (k)$ always equals $\frac{1}{n^k}$, while we have $x_{k+1}(k) = \dots = x_\ell(k)=0$.
Nevertheless, this implies that the acceptance probability of $M^{(\ell)}$ for the string $\sigma^k$ is zero for $k<\ell$, while is
 $|Q_{acc}^{(\ell)}|\cdot x_\ell(\ell)=(\frac{n-1}{n})^{\ell}$ for $k=\ell$. We are now going to prove that for the strings in the language $\sigma^{\ge\ell}$
the acceptance probability never goes below $(\frac{n-1}{n})^{\ell}$. To this aim, it suffices to show
	\begin{thm}
		On the input string $\sigma^{\ell + s}$, with $s\ge 0$, the probability for $M^{(\ell)}$ of being in one of the accepting basis states in $Q_{acc}^{(\ell)}$ while processing the suffix $\sigma^s$ is greater than or equal to $\frac{1}{n^\ell}$.
	\end{thm}
	\begin{proof} We split the proof into two parts, both proved by induction. In the first part,
we focus on the input prefix $\sigma^\ell$. We show by induction on $1 \leq k \leq \ell$ that
$x_h(k) \geq \frac{1}{n^{k}}$ in the system (\ref{eq:rec_general}) holds true for every $1\le h\le k$.
This will enables us to obtain that $x_1(\ell),\ldots,x_\ell(\ell)\ge \frac{1}{n^\ell}$. For the base case $k=1$, 
we recall that $x_1(1)=\frac{1}{n}$. So, let us assume by inductive hypothesis that $x_h(k)\ge\frac{1}{n^k}$ 
for a given $k<\ell$ and every $1\le h\le k$, and prove the property for $k+1$. From the system (\ref{eq:rec_general}), we have
{\small \begin{align*}		
x_h(k+1) =  \frac{1}{n}\cdot x_{h-1}(k) +  \frac{n-1}{n^2}\cdot x_{h}(k)  + \frac{(n-1)^2}{n^3}\cdot x_{h+1}(k) +
		 \dots +\frac{(n-1)^{k-h+1}}{n^{k-h+2}}\cdot x_{k}(k) + \dots + \frac{(n-1)^{\ell-h+1}}{n^{\ell-h+1}}\cdot x_{\ell}(k).
\end{align*}}
		Since $x_{h}(k) \geq \frac{1}{n^k}$ for  $1\le h\leq k$, and $0$ otherwise, we can bound $x_h(k+1)$ from below as
		\begin{align*}
		x_h(k+1) & \geq  \frac{1}{n^k}\cdot\frac{1}{n}\cdot \left(  1 + \frac{n-1}{n} +  \frac{(n-1)^2}{n^2} + \dots + \frac{(n-1)^{k-h+1}}{n^{k-h+1}} \right)
		\geq \frac{1}{n^{k+1}}.
		\end{align*}
		Now, the second part of the proof comes, where we show, again by induction on $k$, that $x_h(k)\geq \frac{1}{n^\ell}$ for $k \geq \ell$ and $1\le h\le\ell$. By the first part of the proof, we have
		$x_1(\ell), x_2(\ell), \dots x_\ell(\ell) \geq \frac{1}{n^\ell}$, and so the base case holds true. 
		We prove $x_h(k+1) \geq \frac{1}{n^\ell}$ assuming such a property for $k$ by inductive hypothesis.
		From the system (\ref{eq:rec_general}), we~get
		{\small \begin{align*}
		x_h(k+1)  = \frac{1}{n}\cdot x_{h-1}(k) +  \frac{n-1}{n^2}\cdot x_{h}(k)  + \frac{(n-1)^2}{n^3}\cdot x_{h+1}(k) + \ldots +
		\frac{(n-1)^{\ell-h}}{n^{\ell-h+1}}\cdot x_{\ell-1}(k) + \frac{(n-1)^{\ell-h+1}}{n^{\ell-h+1}}\cdot x_{\ell}(k). 
		\end{align*}}
		Since we are assuming all $x_j(k)$'s to be greater than or equal to $\frac{1}{n^\ell}$,
		we can bound $x_h(k+1)$ from below~as
		\begin{align*}
		x_h(k+1) \geq  \frac{1}{n^\ell}\cdot  \left(  
		\frac{1}{n} + \frac{n-1}{n^2} +  \frac{(n-1)^2}{n^3} + \dots + \frac{(n-1)^{\ell-h}}{n^{\ell-h+1}} + \left( \frac{n-1}{n}\right)^{\ell-h+1} 
		\right)
		= \frac{1}{n^\ell},
		\end{align*}
		whence, the claimed result follows.
	\end{proof}

	We can conclude that $M^{(\ell)}$ induces the following stochastic event:
		\begin{equation}\label{eq: ml_acc_prob}
	p_{M^{(\ell)}}(\sigma^k) =  |Q^{(\ell)}_{acc}|\cdot x_\ell(k)\; \begin{cases}
		= 0 & \mbox{if $k< \ell$}\\
		\ge \left(\frac{n-1}{n}\right)^\ell & \mbox{if $k\ge \ell$}.
		\end{cases}
	\end{equation} 
This shows that the automaton $M^{(\ell)}$ recognizes $\sigma^{\ge\ell}$ with isolated cut point and $n^{O(\ell)}$ basis states.
As expected, for $n\rightarrow \infty$, the event in (\ref{eq: ml_acc_prob}) approximates a deterministic behavior. In fact, for growing values of $n$, we have $p_{M^{(\ell)}}(\sigma^k)\rightarrow 1$ for $k\ge\ell$, and $p_{M^{(\ell)}}(\sigma^k)=0$ for $k<\ell$.
\smallskip

To sum up, 
let us get back to our initial purpose, i.e., building an isolated cut point \lqfa\ for the language $\sigma^{\ge T}$.
Such a \lqfa\ is obtained by pushing $T$ steps ahead from $M^{(1)}$ the inductive construction to finally get the~\lqfa\ $M^{(T)}$.
As noted, $M^{(T)}$ features $n^{O(T)}$ basis states, $n$ being the number of basis states of $M^{(1)}$. From Equation (\ref{eq: ml_acc_prob}), we can fix a cut point \mbox{$\frac{1}{2}\cdot \left(\frac{n-1}{n}\right)^T$} isolated by $\frac{1}{2}\cdot \left(\frac{n-1}{n}\right)^T$. By increasing~$n$, we widen such an isolation, tending to a deterministic recognition~of~the language $\sigma^{\ge T}$. 
 
Focusing on the size of $M^{(T)}$, we observe that its number of basis states exponentially depends on~$T$.
As a matter of fact, we can avoid such an exponential blow up by noticing that even the \lqfa~$M^{(3)}$
can actually accept with isolated cut point the language $\sigma^{\ge T}$, for $T\ge4$. This is due to the fact that the stochastic event induced by $M^{(3)}$ is an increasing function, as one may readily infer from Equation (\ref{eq: m3_acc_prob}) and Figure~\ref{fig:graph_m3}. By~this property, we can fix the isolated cut point between $p_{M^{(3)}}(\sigma^{T-1})$ and $p_{M^{(3)}}(\sigma^{T})$, thus recognizing~$\sigma^{\ge T}$ with $n^{O(1)}$ basis states, not depending on $T$ any more. Nevertheless, such a dramatic size reduction comes at a price. In fact, the isolation around the
cut point shrinks from $\frac{1}{2}\cdot \left(\frac{n-1}{n}\right)^T$ to $\frac{p_{M^{(3)}}(\sigma^{T})-p_{M^{(3)}}(\sigma^{T-1})}{2}=\frac{1}{2}\cdot (\frac{2}{n})^{T-3}\cdot (\frac{n-1}{n})^{T-1}\cdot(\frac{n+1}{n})$. 
This isolation vanishes as $n$ grows, thus suggesting to consider small values of $n$.
E.g., for $n=2$ we obtain an isolation of $\frac{3}{2}\cdot (\frac{1}{2})^{T}$; for $n=3$ we get $\frac{27}{8}\cdot (\frac{4}{9})^{T}$.

	\section{Isolated Cut Point \lqfas\ for Unary Regular Languages}\label{sec:unary_construction}
	
	Here, we are going to use the \lqfas\ designed in the previous section as
	modules in a more general construction yielding isolated cut point \lqfas\ for unary regular languages.
	This investigation is inspired by \cite{Bianchi10} where the same problem is tackled for \mms. Our result constructively shows that isolated cut point \mms\ and \lqfas\ are equivalent on unary inputs, in sharp contrast to the case for general alphabets where \mms\ outperform \lqfas\ (see Section~\ref{lqfa}).
	 \smallskip

	We start by observing that, according to Theorem \ref{chrobak}, any unary regular language $L \subseteq \sigma^*$ can viewed as the disjoint union of two unary languages, namely, the finite language $L_T=L \cap \sigma^{\le T}$ plus the ultimately periodic language  $L_P = L \cap \sigma^{\ge T+1}$. So, we are going to design two \lqfa\ modules recognizing these two languages with isolated cut point,  and then suitably assemble such modules into a final isolated cut point  \lqfa\ $A_L$ for 	the unary regular language $L$.
\medbreak

\noindent {\bf The finite language $L_T$:}	We define the ``$(T+1)$-periodic continuation''  $L_{T^{\bigcirc}}$ of $L_T$, namely, the language obtained from  $L_T$ by adding all the strings of the form $\sigma^{i+h\cdot (T+1)}$, with $h\ge0$, for $\sigma^i\in L_T$. Formally, $L_{T^{\bigcirc}} = \sett{\sigma^{i+h\cdot (T+1)}}{h\geq 0\mbox{ and }\sigma^i \in L_T}$. Clearly, $L_{T^\bigcirc} $ is a periodic language of period $(T+1)$, and we have that $L_T= L_{T^\bigcirc} \cap \sigma^{\le T}$.
Therefore, in order to recognize $L_{T}$, we start by defining the isolated cut point \lqfa~$A_{T^\bigcirc}$ for $L_{T^\bigcirc}$. We let
$
A_{T^\bigcirc}= (Q, \{\sigma\}, \pi_0, \{U(\sigma), U(\r) \}, \{\O_{\sigma}, \O_{\r}\}, Q_{acc})
$, where:
$Q=\{q_0,\ldots,q_T\}$ is the set of basis states, 	
$Q_{acc}=\sett{q_{i}}{0\le i\le T\mbox{ and }\sigma^i\in L_{T}}$ is the set of accepting basis  states,
$\pi_0 = \e_1$ is the initial superposition,
$U(\sigma) = S$, where $S\in\set{0,1}^{(T+1)\times(T+1)}$ is 
	the matrix representing the cyclic permutation: $S$ has 1 at the $(i, i+1)$th entries for $1\le i\le T$ and at the $(T+1,1)$th entry, all the other entries are 0, $U(\r) = I^{(T+1)} $,
	$\O_{\sigma}$ is the observable having the identity as sole projector,
$\O_\r$ is the usual final observable projecting onto the subspace spanned by $Q_{acc}$.
Given the observable $\O_\sigma$, we have that~$A_{T^\bigcirc}$ is basically a \mo\ whose induced event writes as
$
p_{A_{T^\bigcirc}}(\sigma^k)= \| \pi_0 \cdot U(\sigma)^k \cdot U(\r)\cdot P_{acc}(\r) \|^2.
$ 
After processing the input $\sigma^k\r$, the state $\xi(k)$ of $A_{T^\bigcirc}$ is
\begin{equation}\label{finite_finalsup}
	\xi(k)=\pi_0\cdot U(\sigma)^k\cdot U(\r)=\e_1\cdot U(\sigma)^k\cdot U(\r)= \e_{(k\mod (T+1))+1}.
\end{equation}
Let us now discuss measuring by the final observable, i.e., the action of the projector $P_{acc}(\r)$ on the final superposition $\xi(k)$.
By (\ref{finite_finalsup}), $\xi(k)$ is $\e_{(k\mod (T+1))+1}$, representing the basis state $q_{k\mod (T+1)}$. 
By definition of $Q_{acc}$ we have that $q_{k\mod (T+1)}$ is an accepting state if and only if $\sigma^{k\mod (T+1)} \in L_{T}$ if and only if $\sigma^{k} \in L_{T^{\bigcirc}}$.
Therefore, we can rewrite the stochastic event induced by $A^{T^\bigcirc}$ as $p_{A_{T^\bigcirc}}(\sigma^k) =  \| \xi(k) \cdot P_{acc}(\r) \|^2 = 1$ if  $\sigma^k\in L_{T^{\bigcirc}}$, and $0$ otherwise.
Whence, the \lqfa\ $A_{T^{\bigcirc}}$ recognizes $L_{T^{\bigcirc}}$ by a deterministic event. Now, we need $A_{T^{\bigcirc}}$ to work simultaneously with a module which checks whether or not the input string has length not exceeding $T$, so that the resulting accepted language is $L_{T^\bigcirc} \cap \sigma^{\le T}= L_T$. Such a module can be obtained by complementing
 the \lqfa\ $M^{(T+1)}$ for $\sigma^{\ge T+1}$ presented in Section~\ref{sec:Mh_description} (basically, by taking $Q\setminus Q_{acc}$ as the set of accepting basis states). The resulting  \lqfa\ $\overline{M}^{(T+1)}$ induces the complement of the event in Equation (\ref{eq: ml_acc_prob}) with $\ell=T+1$:
	\begin{equation*}
	p_{\overline{M}^{(T+1)}}(\sigma^k)=1- p_{{M}^{(T+1)}}(\sigma^k) \begin{cases}
		=1 & \mbox{if $k\le T$}\\
		\le 1-\left(\frac{n-1}{n}\right)^{(T+1)} & \mbox{if $k\ge T+1$},
		\end{cases}
	\end{equation*} 
thus recognizing the language $\sigma^{\le T}$ with isolated cut point and $n^{O(T)}$ basis states. Finally, we build the \lqfa\
$A_{T^{\bigcirc}}\otimes \overline{M}^{(T+1)}$ (basically by taking the direct product component wise of the two \lqfas\ $A_{T^{\bigcirc}}$ and~$\overline{M}^{(T+1)}$)
inducing the product event
	\begin{equation*}
	p_{A_{T^{\bigcirc}}\otimes\overline{M}^{(T+1)}}(\sigma^k) =
	p_{A_{T^{\bigcirc}}}\cdot p_{\overline{M}^{(T+1)}}(\sigma^k)  \begin{cases}
		=1 & \mbox{if $\sigma^k\in L_T$}\\
		\le 1-\left(\frac{n-1}{n}\right)^{(T+1)} & \mbox{otherwise},
		\end{cases}
	\end{equation*} 
defining $L_T$ with $(T+1)\cdot n^{O(T)}$ basis states, and cut point $1-\frac{1}{2}\cdot \left(\frac{n-1}{n}\right)^{(T+1)}$ isolated by $\frac{1}{2}\cdot \left(\frac{n-1}{n}\right)^{(T+1)}$. Notice that, for large values of $n$, the \lqfa\ $A_{T^{\bigcirc}}\otimes \overline{M}^{(T+1)}$ approximates a deterministic recognition of~$L_T$.
\medbreak

\noindent {\bf The ultimately periodic language $L_P$:} 	
	It suites our goal to rewrite $L_P$ as $L_P = L_{P^\bigcirc} \cap \sigma^{\ge T+1}$, where we let
	$L_{P^\bigcirc}=\sett{\sigma^{(T+1+i)\mod P\; + h\cdot P}}{0\le i<P,\ h\ge0,\mbox{ and }\sigma^{T+1+i}\in L_P}$. 
	Clearly, $L_{P^\bigcirc}$ is a periodic language of period~$P$. So, for recognizing $L_P$, we first focus on building the isolated cut point \lqfa\ $A_{P^\bigcirc}$ for $L_{P^\bigcirc}$.
We let
	$A_{P^\bigcirc}= (Q,\{\sigma\}, \pi_0, \{U(\sigma), U(\r) \}, \{\O_\sigma, \O_\r\}, Q_{acc})$, where:
$Q=\{q_0,\ldots,q_{P-1}\}$ is the set of basis states,		
$Q_{acc}=\sett{q_{i}}{0\le i<P\mbox{ and }\sigma^i\in L_{P^{\bigcirc}}}$ is the set of accepting basis  states,	
$\pi_0 = \e_1$ is the initial superposition,	
$U(\sigma) = S$, where $S\in\set{0,1}^{P\times P}$ is the cyclic permutation matrix, $U(\r)=I^{(P)}$,		
$\O_{\sigma}$ is the observable having the identity as sole projector,	
$\O_\r$ is the usual final observable projecting onto the subspace spanned by $Q_{acc}$.
	Given the observable $\O_\sigma$, we have that $A_{P^\bigcirc}$ is basically a \mo\ whose induced event writes as
	$
	p_{A^{P^\bigcirc}}(\sigma^k)= \| \pi_0 \cdot U(\sigma)^k \cdot U(\r)\cdot P_{acc}(\r) \|^2.
	$ 
	After processing the input $\sigma^k\r$, the state $\xi(k)$ of $A_{P^\bigcirc}$ is
	\begin{equation}\label{finalsupP}	
		\xi(k)=\pi_0\cdot U(\sigma)^k\cdot U(\r)=\e_1\cdot U(\sigma)^k\cdot U(\r)=	\e_{(k\mod P)+1}.
	\end{equation}
Let us now measure the final observable on the final superposition $\xi(k)$. 
By (\ref{finalsupP}), $\xi(k)$ is $\e_{(k\mod P)+1}$, representing the basis state $q_{k\mod P}$. By definition of 
$Q_{acc}$, we have that $q_{k\mod P}$ is an accepting state if and only if $\sigma^k\in L_{P^\bigcirc}$.
Therefore,
the stochastic event induced by $A_{P^\bigcirc}$ is $p_{P^\bigcirc}(\sigma^k) =  \left\| \xi(k) \cdot P_{acc}(\r) \right\|^2  = 1$, 
if  $\sigma^k\in L_{P^\bigcirc}$, and $0$ otherwise.
whence, the \lqfa\ $A_{P^\bigcirc}$ recognizes $L_{P^\bigcirc}$ by a deterministic event. Now, we need $A_{P^{\bigcirc}}$ to work simultaneously with a module which checks whether or not the input string has length exceeding $T$, so that the resulting accepted language is $L_{P^\bigcirc} \cap \sigma^{\ge T+1}= L_P$. Such a module is the \lqfa\ $M^{(T+1)}$ for $\sigma^{\ge T+1}$ presented in Section \ref{sec:Mh_description}, and inducing the event
	\begin{equation*}
	p_{{M}^{(T+1)}}(\sigma^k) \begin{cases}
		\ge (\frac{n-1}{n})^{T+1} & \mbox{if $k\ge T+1$}\\
		= 0 & \mbox{$k\le T$},
		\end{cases}
	\end{equation*} 
thus recognizing the language $\sigma^{\ge T+1}$ with isolated cut point and $n^{O(T)}$ basis states. Finally, we build the \lqfa\
$A_{P^{\bigcirc}}\otimes {M}^{(T+1)}$,
inducing the product event
	\begin{equation*}
	p_{A_{P^{\bigcirc}}\otimes {M}^{(T+1)}}(\sigma^k) =
	p_{A_{P^{\bigcirc}}}\cdot p_{{M}^{(T+1)}}(\sigma^k)  \begin{cases}
		\ge (\frac{n-1}{n})^{T+1} & \mbox{if $\sigma^k\in L_P$}\\
		=0 & \mbox{otherwise},
		\end{cases}
	\end{equation*} 
defining $L_P$ with $P\cdot n^{O(T)}$ basis states, and cut point $\frac{1}{2}\cdot \left(\frac{n-1}{n}\right)^{(T+1)}$ isolated by $\frac{1}{2}\cdot \left(\frac{n-1}{n}\right)^{(T+1)}$. Notice that, for large values of $n$, the \lqfa\ $A_{P^{\bigcirc}}\otimes {M}^{(T+1)}$ approximates a deterministic recognition of $L_P$.
\medbreak

\noindent {\bf Putting things together:}
We are now ready to suitably assemble the two \lqfas\ $A_T=A_{T^{\bigcirc}}\otimes \overline{M}^{(T+1)}$ and $A_P=A_{P^{\bigcirc}}\otimes {M}^{(T+1)}$ so far described 
to obtain an isolated cut point \lqfa\ $A_L$ for the unary regular language $L$.
We notice that $L=L_T\cup L_P=(L_T^c\cap L_P^c)^c$. This suggests first to construct \lqfas\ for $L_T^c$ and~$L_P^c$
by building $\overline{A}_{T}$ and $\overline{A}_{P}$ inducing the complement events $p_{\overline{A}_{T}}=1-p_{A_{T}}$ and $p_{\overline{A}_{P}}=1-p_{A_{P}}$, respectively. Next, to account for the intersection, we construct the \lqfa\ $\overline{A}_L=\overline{A}_{T}\otimes \overline{A}_{P}$ inducing the product event $p_{\overline{A}_{L}}=(1-p_{A_{T}})\cdot (1-p_{A_{P}})$. Finally, the desired \lqfa\ $A_L$ will be obtained by complementing $\overline{A}_L$, so that 
$p_{A_L}=(1-p_{\overline{A}_L})=1-(1-p_{A_{T}})\cdot (1-p_{A_{P}})=p_{A_{T}}+p_{A_{P}}-p_{A_{T}}\cdot p_{A_{P}}.$

Let us now explain how $p_{A_L}$ behaves on input string $\sigma^k$:
\begin{itemize}

\item {$\sigma^k\in L=L_T\cup L_P$}:\ \ Clearly, we have either $\sigma^k\in L_T$ or $\sigma^k\in L_P$. 
Suppose $\sigma^k\in L_T$. Then, we have that $p_{A_{T}}(\sigma^k)=1$ since $A_T=A_{T^{\bigcirc}}\otimes \overline{M}^{(T+1)}$ and both its sub-modules will accept with certainty; correspondingly, $p_{A_{P}}(\sigma^k)=0$ since $A_P=A_{P^{\bigcirc}}\otimes {M}^{(T+1)}$ and the sub-module ${M}^{(T+1)}$ accepts with~0 probability the input strings of length less than or equal~to~$T$. Globally, we have $p_{A_{L}}(\sigma^k)=1$.
Suppose $\sigma^k\in L_P$. Then, we have that $p_{A_{P}}(\sigma^k)\ge \left(\frac{n-1}{n}\right)^{T+1}$ since 
$A_{P^{\bigcirc}}$ accepts with certainty, while the sub-module ${M}^{(T+1)}$ accepts with probability not less than $\left(\frac{n-1}{n}\right)^{T+1}$. Let us now focus on~$A_T$. The sub-module $A_{T^{\bigcirc}}$ could accept with probability either 0 or 1. In the former case, globally we have $p_{A_{L}}(\sigma^k)\ge \left(\frac{n-1}{n}\right)^{T+1}$, in the latter, the sub-module $\overline{M}^{(T+1)}$ accepts with a probability bounded above by $1-\left(\frac{n-1}{n}\right)^{T+1}$. By letting $(1-y)$ the acceptance probability of $\overline{M}^{(T+1)}$, with $0\le y\le \left(\frac{n-1}{n}\right)^{T+1}$, we get
$p_{A_{L}}(\sigma^k)\ge\left(\frac{n-1}{n}\right)^{T+1}+(1-y)-\left(\frac{n-1}{n}\right)^{T+1}\cdot(1-y)\ge\left(\frac{n-1}{n}\right)^{T+1}.$

In conclusion, for any $\sigma^k\in L$, we have $p_{A_{L}}(\sigma^k)\ge\left(\frac{n-1}{n}\right)^{T+1}.$
 
 	\item {$\sigma^k\not\in L=L_T\cup L_P$}:\ \ Clearly, both $\sigma^k\not\in L_T$ and $\sigma^k\not\in L_P$.
By assuming $k\le T$, we must have $\sigma^k\not\in L_{T^\bigcirc}$. 	Therfore, the sole acceptance probability contribution could come from the module $A_P=A_{P^{\bigcirc}}\otimes {M}^{(T+1)}$. However, since $k\le T$, the sub-module ${M}^{(T+1)}$ accepts with 0 probability. 
So, $p_{A_L}(\sigma^k)=0$.
Instead, by assuming $k\ge T+1$, we must have that $\sigma^k\not\in L_{P^\bigcirc}$. Thus, the sole acceptance probability could come from the module $A_T$. However, the acceptance probability yielded by the sub-module~$\overline{M}^{(T+1)}$ turns out to be at most $1-\left(\frac{n-1}{n}\right)^{T+1}$. 

In conclusion, for any $\sigma^k\not\in L$, we have $p_{A_L}(\sigma^k)\le 1-\left(\frac{n-1}{n}\right)^{T+1}$.
	\end{itemize}
Summing up, the stochastic event induced by the \lqfa\ $A_L$ is
\begin{equation}\label{evL}
p_{A_L}(\sigma^k)\begin{cases}
\ge\left(\frac{n-1}{n}\right)^{T+1} & \text{if $\sigma^k\in L$}\\
\le 1-\left(\frac{n-1}{n}\right)^{T+1} & \text{otherwise}.
\end{cases}
\end{equation}
By the event in Equation (\ref{evL}), we get that $A_L$ recognizes $L$ with the following cut point and isolation~radius:
{\small
$$\lambda=\frac{1}{2}\cdot\left(\left(\frac{n-1}{n}\right)^{T+1}\!\!+1-\left(\frac{n-1}{n}\right)^{T+1}\right)=\frac{1}{2},\ \ \
\varrho=\frac{1}{2}\cdot\left(\left(\frac{n-1}{n}\right)^{T+1}\!\!\!-1+\left(\frac{n-1}{n}\right)^{T+1}\right)= \left(\frac{n-1}{n}\right)^{T+1}\!\!-\frac{1}{2}.$$}

\noindent Clearly, to have an isolation around $\lambda$, we must require that $\rho>0$. This can always be achieved on any  $T>0$ by imposing $\left(\frac{n-1}{n}\right)^{T+1}>\frac{1}{2}$, which is attained whenever 
$n>\frac{1}{1-\sqrt[T+1]{\frac{1}{2}}}$. This latter condition is satisfied, e.g., by letting $n=4\, T$ for any $T>0$. Nevertheless, the isolation radius $\rho$ tends to $\frac{1}{2}$ as $n$~grows. 
\medbreak

Let us 
inspect the size of the \lqfa\ $A_L=\overline{\overline{A}_T\otimes\overline{A}_P}$. As above pointed out, $A_T$ and $A_P$ have, respectively, $(T+1)\cdot n^{O(T)}$ and $P\cdot n^{O(T)}$ basis states. 
The complements $\overline{A}_T$ and $\overline{A}_P$ maintain the same number of basis states, while
the product $\overline{A}_T\otimes\overline{A}_P$ requires $((T+1)\cdot n^{O(T)})\cdot(P\cdot n^{O(T)})\le T\cdot P\cdot n^{O(T)}$ basis states. The final complement $\overline{\overline{A}_T\otimes\overline{A}_P}$ maintains the same number of basis states. By replacing $n$ with $4\,T$, as above suggested, the number of basis states of the isolated cut point \lqfa\ $A_L$ for $L$ becomes~$P\cdot T^{O(T)}$.

\section{Conclusions}\label{sec:conclusion}

In  this work, we have exhibited a modular framework for building isolated cut point \lqfas\ for unary regular languages. 
By suitably adapting to the unary case an inductive construction in \cite{Ambainis06,Mercer07}, we have first designed \lqfas\ discriminating unary inputs on the basis of their length.
These devices have then been plugged into two sub-modules recognizing the finite part and the ultimately periodic part
any unary regular language consists of. The resulting \lqfa\ recognizes a unary regular language $L$ with isolated  cut point~$\frac{1}{2}$, and a number of basis states which is exponential in the number of states of the minimal \dfa\ for~$L$. In spite of this exponential size blow up, it should be stressed that more restricted models of quantum finite automata in the literature, such as \mos, cannot recognize all unary regular languages. On the other hand, a linear amount of basis states is sufficient for the more powerful model of isolated cut point  \mms\ \cite{Bianchi10}. Thus, it would be worth investigating whether a more size efficient construction for unary \lqfas\ could be provided. 
Another interesting line of research might explore the descriptional power (see, e.g., \cite{BGMP17,HK11,KMMP20,KMMP20a} for topics in descriptional complexity)  of isolated cut point \lqfas\ with respect to other relevant classes of subregular languages such as, e.g.,~commutative~regular~languages~\cite{bianca}.
\medbreak

\noindent {\bf Acknowledgements.} The authors wish to thank the anonymous referees for their valuable comments.

\nocite{*}
\bibliographystyle{eptcs}
\bibliography{biblio}
\end{document}